\newcommand{\BE}{\begin{eqnarray}}
\newcommand{\EE}{\end{eqnarray}}
\newcommand{\be}{\begin{eqnarray}}
\newcommand{\ee}{\end{eqnarray}}
\newcommand{\ben}{\begin{eqnarray*}}
\newcommand{\een}{\end{eqnarray*}}
\newcommand{\ba}{\begin{array}}
\newcommand{\ea}{\end{array}}
\newcommand{\BLN}{\begin{align*}}
\newcommand{\ELN}{\end{align*}}

\newcommand{\m}{\mbox}

\newcommand{\cd}{\cdot}
\newcommand{\f}{\frac}

\newcommand{\bl}{\begin{lstlisting}}
\newcommand{\el}{\end{lstlisting}}

\newcommand\disp{\displaystyle}

\newcommand{\fig}[1]{Figure \ref{#1}}
\newcommand{\Ai}{\operatorname{Ai}}
\newcommand{\Li}{\operatorname{Li}}

\documentclass[cmp]{svjour}
\usepackage{graphics}
\usepackage[latin1]{inputenc}
\usepackage{subfigure}
\usepackage[T1]{fontenc}
\usepackage{graphicx}
\usepackage{latexsym}
\usepackage{amsmath,amssymb}
\usepackage{ifpdf}
\usepackage{dsfont} % fuer Einheitsmatrixzeichen
\usepackage{cancel} % Durchstreichen von Ausdruecken
\usepackage{epsfig}
\usepackage{subfigure}
\usepackage{amsmath,amssymb,color}
\usepackage[english]{babel}
\begin{document}
\title{Uniform asymptotics of area-weighted Dyck paths}
\author{Nils Haug\inst{1} \and Thomas Prellberg\inst{1}%
}                     % Do not remove
\institute{School of Mathematical Sciences, Queen Mary University of London, London, E1 4NS}
\date{Received: date / Accepted: date}

\communicated{}
\maketitle
\begin{abstract}
Using the generalized method of steepest descents for the case of two coalescing saddle points, we derive an asymptotic expression for the bivariate generating function of Dyck paths, weighted according to their length and their area in the limit of the area generating variable tending towards 1. The result is valid uniformly for a range of the length generating variable, including the tricritical point of the model.
\end{abstract}
\section{Introduction}
\label{introduction}

A Dyck path is a trajectory of a directed random walk on the two-dimensional square lattice above the diagonal $y=x$, starting at the origin and ending on this diagonal. More precisely, the random walk starts at the point $(0,0)$ and, from any given point $(x,y)$, the random walker can only step towards $(x+1,y)$ and $(x,y+1)$. Steps from $(x,y)$ to $(x-1,y)$ or to $(x,y-1)$ are forbidden. Furthermore, each point on the trajectory $(x,y)$ must satisfy $x\leq y$, which means that the random walker always stays above the main diagonal $x=y$, and its final position $(x,y)$ must be on the main diagonal. \fig{fig:dp_01} shows an example of a Dyck path as it is usually drawn, with the lattice oriented such that the main diagonal lies horizontally in the image. As a possible physical system which can be described by Dyck paths, one can think of two-dimensional vesicles attached to a surface \cite{Prellberg13}, or of discrete trajectories of charged particles moving in an external magnetic field \cite{Bulycheva14}.\\
Since we will only deal with Dyck paths in the following, we will refer to them simply as ``paths''.

\begin{figure}[ht]
\centering \includegraphics[width=0.75\textwidth]{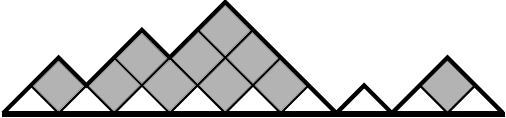}
\caption{A Dyck path of area $m=10$ and length $n=18$. The shaded squares have unit area. The white triangles on the bottom do not contribute to the area.}
\label{fig:dp_01}
\end{figure}

The \emph{length} of a path is the number of steps it consists of. Since the number of horizontal steps must equal the number of vertical steps, it follows that the length of a path is always even. Further, the \emph{area} of a path is the number of entire unit squares enclosed between the trajectory and the main diagonal. For example, the path in \fig{fig:dp_01} has length $n=18$ and area $m=10$ (occasionally an alternative definition of the area under a path as the number of triangular plaquettes enclosed by the trajectory and the main diagonal is used, see e.g. \cite{Flajolet80}).\\

The generating function of paths of length $2n$, weighted according to their area is defined as 
\be 
	Z_n(q) = \sum_{m=1}^\infty c_{m,n}\,q^m,
\label{eq:dp_01}
\ee 

\noindent where $c_{m,n}$ is the number of paths of length $2n$ and area $m$ and $q$ is the weight associated to the area. In physical terms, one can interpret $Z_n(q)$ as the canonical partition function of two-dimensional surface-attached vesicles with perimeter $4n$ and area-fugacity $q$. Note that the maximum area of a path of given finite length is bounded, therefore the sum on the RHS of \eqref{eq:dp_01} is finite for finite values of $n$. As a physical consequence, phase transitions can only occur in the thermodynamic limit $n\to\infty$.

One can also define the generating function of paths of area $m$, weighted with respect to their length, as 
\be 
Q_m(t) = \sum_{n=0} c_{m,n}\,t^n,
\label{eq:dp_015}
\ee
where $t$ is the weight conjugate to the length. The generating function of paths weighted according to both their area \emph{and} their length is defined as
\be 
	G(t,q) = \sum_{n=1}^\infty\sum_{m=1}^\infty c_{m,n}\,q^m\,t^n.
\label{eq:dp_02}
\ee 

\noindent With \eqref{eq:dp_01}, this can be rewritten as
\be 
G(t,q) = \sum_{n=1}^\infty Z_n(q)\,t^n,
\label{eq:dp_03}
\ee 
and with \eqref{eq:dp_015}, we can write 
\be 
G(t,q) = \sum_{m=1}^\infty G_m(t)\,q^m.
\ee
By a standard factorization argument \cite{Flajolet09}, one obtains the functional equation 
\be 
G(t,q) = 1 + t\,G(t,q) G(qt,q),
\label{eq:dp_05}
\ee 
which can be solved by using the ansatz
\be 
G(t,q) = \f{H(qt)}{H(t)}.
\label{eq:dp_07}
\ee
Here,
\be 
H(t) = \sum_{n=0}^\infty \f{q^{n^2-n}(-t)^n} {(q;q)_n},
\label{eq:dp_08}
\ee
and we have used the standard notation for the q-Pochhammer symbol,
\be 
(z;q)_n = \prod_{k=0}^{n-1} (1-zq^k),
\label{eq:dp_09}
\ee 
which is a q-generalization of the Pochhammer symbol. The function $H(qt)=\Ai_q(t)$ is a q-Airy function \cite{Ismail05}.

\noindent For $q = 1$, we obtain the generating function of the Catalan numbers (see e.g. \cite{Aigner07}),
\be 
G(t,1) = \f{1}{2t}(1-\sqrt{1-4t}).
\label{catalangf}
\ee 

\begin{figure}[ht]
\centering \includegraphics[width=0.35\textwidth]{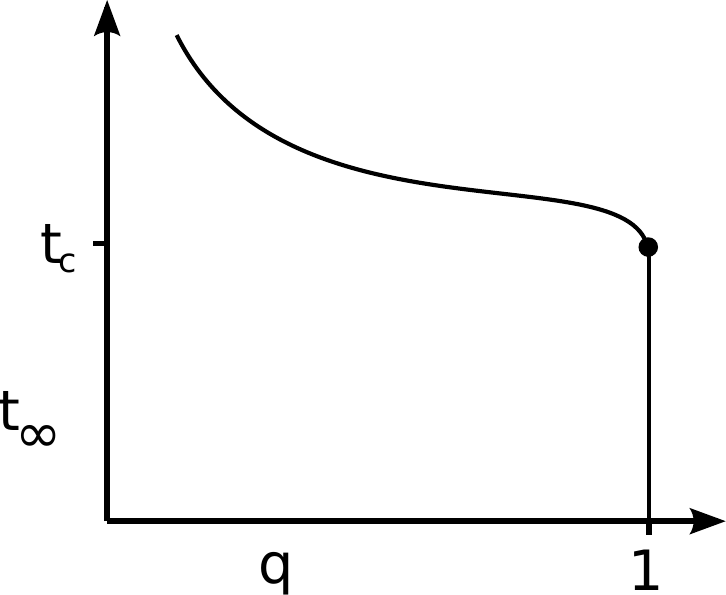}
\caption{The qualitative behaviour of the radius of convergence $t_\infty$ of $G(t,q)$ as a function of $q$. The small circle marks the tricritical point.}
\label{fig:dp_02}
\end{figure}

In \fig{fig:dp_02}, we show how the radius of convergence $t_\infty$ of the series on the RHS of Eq.\eqref{eq:dp_03} behaves qualitatively as a function of $q$. This picture, which is also called the ``phase diagram'' of the system, is typical for lattice polygon models \cite{Richard09}. The radius of convergence is defined by a decreasing line of pole singularities for $q<1$ and zero for $q>1$. From (\ref{catalangf}), the value of $t_\infty$ for $q=1$ can be deduced to be $t_c=1/4$. The point $(t,q)=(t_c,1)$ is called the \emph{tricritical point} of the model (see e.g. \cite{VanRensburg00}) and the area below $t_\infty(q)$ is called the \emph{finite size region}.\\

In \cite{Flajolet80}, the general form of continued fraction expressions for generating functions of Dyck and Motzkin paths\footnote{Motzkin paths are a generalisation of Dyck paths.} has been discussed. In particular, 
\be 
G(t,q) = \f{1}{1-\disp{\f{t}{1-\disp{\f{tq}{1-\disp{\f{tq^2}{1-\disp{\f{tq^3}{1-\dots}}}}}}}}}.
\ee
This expression enables us to continue $G(t,q)$ analytically beyond the finite size region.\\

The asymptotic behaviour of $G(t,q)$ for $q\to 1^-$ as one approaches the tricritical point has so far not been derived rigorously. The aim of this paper is to close this gap by rigorously deriving an asymptotic expression for the generating function $G(t,q)$ in the limit $q\to 1^-$ which is valid uniformly for a range of values of $t$ including the critical point $t_c$. A similar calculation has been carried out in \cite{Prellberg95} for staircase polygons.\\

\noindent Our main result is given in Proposition \ref{prop:dp_01} and Corollary \ref{Corollary:dp_01}.\\

\section{Results}
\label{results}

According to Eq.\eqref{eq:dp_07}, the function $G(t,q)$ is given as a quotient of two alternating q-series. In order to obtain its asymptotic behaviour, we first derive the asymptotic behaviour of both the enumerator and the denominator separately. Taking the fraction of the two obtained expressions will then lead us to the asymptotic behaviour of $G(t,q)$.  We will start with the asymptotic expansion of $H(t)$.

\subsection{Uniform asymptotic expansion of $H(t)$}
\label{asymptotic_H(t)}

The first step in our calculation is to express $H(t)$ as a contour integral. We prove

\begin{lemma}{}

For complex $t$ and $0 < q < 1$,
\BE
H(t) = \f{(q;q)_\infty}{2\pi i}\int_{C}\f{z^{(1+\log_q z)/2-\log_q t}}{(z;q)_\infty} dz,
\label{eq:dp_12}
\EE
where $C$~is a contour as shown in \fig{fig:dp_03}b.

\label{Lemma:dp_01}
\end{lemma}

\begin{figure}[htbp]
\centering   \subfigure[]{
    \label{fig:dp_03a}
\includegraphics[width=0.3\textwidth]{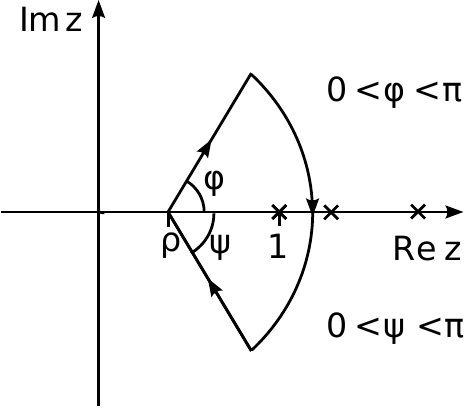}
  }\hspace{1.5cm}
\centering   \subfigure[]{
    \label{fig:dp_03b}
\centering \includegraphics[width=0.3\textwidth]{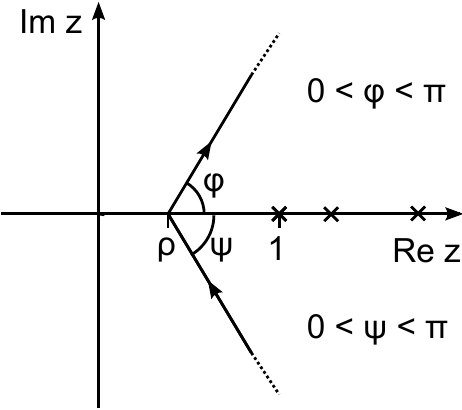}
  }
  \caption{The contours $C_1$ (a) and $C$ (b).}
  \label{fig:dp_03}
\end{figure}

\begin{proof}
For complex $q \neq 0$ and $n\in \mathbb{N}_0$, we have
\be
\disp \f{(-1)^nq^{n\choose 2}}{(q;q)_n(q;q)_\infty} = -\m{Res}\,\left[(z;q)_\infty^{-1};z=q^{-n}\right],
\label{eq:dp_13}
\ee
from which it follows that
\be
\f{(-t)^n q^{n^2-n}}{{(q;q)_n(q;q)_\infty}}
=-\m{Res}[\frac{z^{(1+\log_q z)/2-\log_q t}}{(z;q)_\infty};z=q^{-n}].
\label{eq:dp_14}
\ee
Suppose now that $0<q<1$. Then the contour $C_N = C_N^1 \cup C_N^2 \cup C_N^3$, where
\be \left.
\begin{array}{rcl}
C_N^1&=&\big\{ \rho + \lambda e^{-i\psi}~|~0 < \lambda < q^{-N-1/2}\big\},\vspace*{2mm}\\
C_N^2&=&\big\{ \rho + \lambda e^{i\varphi}~|~0 < \lambda < q^{-N-1/2}\big\},\vspace*{2mm}\\
C_N^3&=&\big\{ \rho + q^{-N-1/2} e^{i \theta}~|~-\psi < \theta < \varphi \big\}
\end{array}
\label{C_N}
\right\},
\ee 
$0<\rho<1$ and $(\varphi,\psi)\,\in~]\,0,\pi\,[^2$, surrounds exactly the $N$ leftmost singularities of the integrand on the RHS of Eq.\eqref{eq:dp_12} -- see \fig{fig:dp_03}a . We can therefore write
\be 
\sum_{n=0}^N \f{q^{n^2-n}(-t)^n} {(q;q)_n} = \f{(q;q)_\infty}{2\pi i}\oint_{C_N} \f{z^{(1+\log_q z)/2-\log_q t}}{(z;q)_\infty} dz.
\label{eq:dp145}
\ee
where the integration is performed in clockwise sense, as indicated by the arrows in \fig{fig:dp_03}a. Combining \eqref{eq:dp_08} and \eqref{eq:dp145}, we obtain
\BE
H(t) = \disp \lim_{N\to\infty} \f{(q;q)_\infty}{2\pi i}\oint_{C_N} \f{z^{(1+\log_q z)/2-\log_q t}}{(z;q)_\infty} dz.
\label{eq:dp_15}
\EE

It is left to show that in the limit $N\to\infty$, the contribution of the circle segment $C_N^3$ to the contour integral (\ref{eq:dp_15}) vanishes, such that the contour $C_N$ can be replaced by the contour shown in \fig{fig:dp_03}b.\\

\noindent On $C_N^3$, we can estimate the denominator of the integrand on the RHS of \eqref{eq:dp_15} as
\begin{multline}
\big|(z,q)_\infty\big| = \disp \Big|\prod_{n=0}^\infty (1-q^{-N-1/2+n}e^{i\varphi})\Big| \geq \disp \Big|\prod_{n=0}^\infty (1-q^{-N-1/2+n})\Big| \\
= \disp \Big|\prod_{n=0}^\infty (1-q^{-1/2+n})\Big|\cd \Big|\prod_{n=1}^{N} (1-q^{-1/2-n})\Big| \geq c_1~\disp \Big|\prod_{n=1}^{N} q^{-1/2-n}\Big|
= \disp c_1\,\big| q^{-N^2/2-N}\big|,
\label{eq:dp_16}
\end{multline}
where $c_1$~is a constant independent of $N$. \\
\noindent Furthermore, the absolute value of the enumerator has for $z\in C_N^3$ the upper bound
\BE
|z^{(1+\log_q z)/2-\log_q t}| \leq c_2~|q^{N^2/2}| |t|^{N},
\label{eq:dp_17}
\EE
where $c_2$ is another constant independent of $N$. Therefore, we can estimate
\be
		\f{1}{2\pi i}\oint_{C_N^3} \f{z^{(1+\log_q z)/2-\log_q t}}{(z;q)_\infty} dz \leq c_3\,|q|^{N^2}|t|^{N},
\label{eq:dp_18}
\ee
where $c_3$ is a third constant independent of $N$. Since the expression on the right hand side tends to zero as $N\to\infty$ for $q<1$, it follows that
\be 
\lim_{N\to\infty} \f{1}{2\pi i}\oint_{C_N} \f{z^{(1+\log_q z)/2-\log_q t}}{(z;q)_\infty} dz=\f{1}{2\pi i}\int_{C}\f{z^{(1+\log_q z)/2-\log_q t}}{(z;q)_\infty} dz,
\label{eq:dp_19}
\ee
where 
\be 
C = \big\{ \rho + \lambda e^{-i\psi}~|~0<\lambda<\infty\big\}\cup\big\{ \rho + \lambda e^{i\varphi}~|~0<\lambda<\infty\big\},
\label{eq:C}
\ee
and where the integration is carried out as indicated by the arrows in \fig{fig:dp_03}b. Combining (\ref{eq:dp_19}) with (\ref{eq:dp_15}), we obtain (\ref{eq:dp_12}).

\end{proof}

\noindent In \cite{Prellberg95} it was shown by applying the Euler-Maclaurin summation formula that
\be 
\ln (z;q)_\infty = -\f{1}{\ln(q)}\Li_2(z)+\f{1}{2}\ln(1-z)+\ln(q)\,R(z,q),
\label{eq:dp_20}
\ee 
where the remainder satisfies the non-uniform bound
\be 
|R(z,q)| \leq \f{1}{6}\left(\ln|1-z|+\f{\m{Re}(z)}{\m{Im}(z)}\arctan\f{\m{Im}(z)}{1-\m{Re}(z)}\right).
\label{eq:dp_21}
\ee 
Here, $\Li_2$ denotes the Euler dilogarithm \cite{NIST}, which can be defined as
\be 
\Li_2(z) = -\int_0^z \f{\ln(1-s)}{s}ds.
\ee
Combining \eqref{eq:dp_12} with \eqref{eq:dp_20} and defining
\begin{subequations}
\BE
	f(z,t)&=&\ln(t)\ln(z)+\Li_2(z)-\f{1}{2}\ln(z)^2,\vspace*{2mm}\label{eq:dp_17a}\\	
	g(z)&=&\disp \left(\f{z}{1-z}\right)^{1/2},\label{eq:dp_17b}
	\label{eq:dp_23}
\EE
\end{subequations}
we can write
\BE
H(t) = \f{(q;q)_\infty}{2\pi i}\int_{C} \exp\left(\f{1}{\epsilon}f(z,t)+\epsilon R(z,q) \right)g(z)dz,
\label{eq:dp_22}
\EE
where $\epsilon=-\ln(q)$. \\

The function $f(z,t)$ is analytic for $\arg(z)<\pi$ and $\arg(1-z)<\pi$ and has the two saddle points
\be 
	z_1(t) = \f 12\left(1+\sqrt{1-4t}\right)\quad;\quad z_2(t) = \f 12\left(1-\sqrt{1-4t}\right)~,
	\label{eq:dp_24}
\ee 
which coalesce for $t=t_c= 1/4$.\\

\noindent From the identity
\be 
\Li_2(\lambda e^{i\phi}) = -\f{1}{2}\ln(-\lambda e^{i\phi})^2-\f{\pi^2}{6} -\Li_2(\f{1}{\lambda}e^{-i\phi})
\label{eq:dp_26}
\ee
(see e.g. \cite{Maximon03}), we obtain
\begin{lemma}

For $0 < |\phi| \leq \pi$,
\begin{subequations}
\be
\Li_2(\lambda e^{i\phi})) &\sim & -\f{1}{2}\ln(-\lambda e^{i\phi})^2,\label{eq:dp_25a}\\
\emph{Re}[\Li_2(\lambda e^{i\phi}))] &\sim &  -\f{1}{2}\ln(\lambda)^2,\label{eq:dp_25b}\\
\emph{Im}[\Li_2(\lambda e^{i\phi}))] &\sim &  -\f{1}{2}\emph{Im}[(\ln(-\lambda e^{i\phi})^2]\label{eq:dp_25c}
\label{eq:dp_25}
\ee
\end{subequations}
as $\lambda \to \infty$.
\label{Lemma:dp_02}
\end{lemma}
\noindent Consequently we have
\begin{corollary}

For complex $t$ and $0 < |\phi| < \pi$,
\be 
f(\lambda e^{i\phi},t) \sim -\ln(\lambda)^2-i\psi \ln(\lambda) \quad\emph{as}~\lambda\to\infty,
\label{eq:dp_27}
\ee
where $\psi = 2\phi+\pi$ for $\phi < 0$ and $\psi = 2\phi-\pi$ for $\phi>0$. %where the signs have to be chosen such that $-\pi <2\phi\pm \pi < \pi$.

\label{Corollary:dp_01}

\end{corollary}
The remainder $R(z,q)$ is not uniformly bounded with respect to $z$, therefore it is not immediately clear that it can be neglected in the limit $\epsilon\to 0^+$. However, from the asymptotic behaviour of $f(z,t)$ one can conclude that the tails of the integration contour do not contribute to the asymptotics of the integral. Therefore we have
\begin{lemma}{}

For complex $t$ and $0<q<1$,
\BE
H(t) \sim \f{(q;q)_\infty}{2\pi i} \int_{C} \exp\left(\f 1\epsilon f(z,t) \right) g(z) dz\qquad\m{as}~\epsilon \to 0^+.
\label{eq:dp_28}
\EE

\label{Lemma:dp_03}
\end{lemma}

It was shown in \cite{Chester57} (see also \cite{Bleistein75} and \cite{Wong01}) that for a function $f(z,t)$, which is analytic with respect to both $z$ and $t$ and which has two saddle points $z_1(t)$ and $z_2(t)$, there is a unique transformation $u:z \longmapsto u(z)$, such that
\BE
	f(z,t) = \f{1}{3}u^3-\alpha(t) u+\beta(t),
	\label{eq:dp_37}
\EE 
which is regular and one-to-one in a domain containing $z_1(t)$ and $z_2(t)$ if $t$ lies in some small domain containing $t_c$. Moreover,
\be 
	u\big(z_1(t)\big) = |\alpha(t)^{1/2}|\quad;\quad u\big(z_2(t)\big) = -|\alpha(t)^{1/2}|.
	\label{eq:dp_38}
\ee
Combining \eqref{eq:dp_37} with \eqref{eq:dp_38}, one gets the explicit form
\begin{multline}
u(z) = \left[\left(\f 32 (f(z,t)-\beta)\right)^2+\left(\left(\f 32 (f(z,t)-\beta)\right)^2-\alpha^3\right)^{1/2}\right]^{1/3}+ \\
+ \alpha\left[\left(\f 32 (f(z,t)-\beta)\right)^2+\left(\left(\f 32 (f(z,t)-\beta)\right)^2-\alpha^3\right)^{1/2}\right]^{-1/3},
\label{eq:dp_39}
\end{multline}
where the two parameters are obtained by inserting \eqref{eq:dp_38} into \eqref{eq:dp_37} as
\be 
\left. \begin{array}{rcl}
\disp \alpha(t) & = & \disp \left(\f 34 \left[ f(z_2,t)-f(z_1,t) \right] \right)^{2/3}~,\vspace{2.5mm}\\
\disp \beta(t)  & = & \disp \f 12  \left( f(z_1,t) + f(z_2,t) \right) = \f 12 \ln(t)^2+\f {\pi^2}{6}
\end{array}\right\}.
\label{eq:dp_40}
\ee
Note that for $t\to t_c$, one has $\alpha(t)\sim 1-4t$. From Corollary \ref{Corollary:dp_01} it follows that
\be 
	u(\rho + i\lambda ) \sim \exp\left(\pm i \f{\pi}{3}\right)\left(\f 32 \ln |\lambda| \right)^{1/3},
	\label{eq:dp_41}
\ee
for $\lambda\to\pm\infty$. This leads us to
\begin{lemma}{}
For complex $t$,
\BE
H(t) \sim \f{(q;q)_\infty}{2\pi i}  \int_{C} e^{\disp\f 1\epsilon \left(\f 13 u^3-\alpha u+\beta \right)} g(z(u)) \f{dz}{du} du
\label{eq:dp_42}
\EE
as $\epsilon \to 0^+$, where $C$ is a contour as shown in \fig{fig:dp_03}b with $\varphi=\psi=\pi/3$.

\label{Lemma:dp_04}
\end{lemma}
\noindent It is possible to write
\be 
	g(z(u))\f{dz}{du} = \sum_{m=0}^\infty (p_m+u q_m)(u^2-\alpha)^m
	\label{eq:dp_43}
\ee
and insert this expansion into \eqref{eq:dp_42}. Interchanging the order of integration and summation, we obtain the asymptotic expansion
\be 
H(t) \sim \f{(q;q)_\infty}{2\pi i}  \sum_{m=0}^\infty \int_{C} (p_m+u q_m)(u^2-\alpha)^m e^{\disp\f 1\epsilon \left(\f 13 u^3-\alpha u+\beta \right)} du.
\label{eq:H_series}
\ee
The two leading coefficients can be obtained from \eqref{eq:dp_37} and \eqref{eq:dp_38}. We get
\be 
2\,g(z_1)\sqrt{\f{\alpha}{f^{\prime\prime}(z_{1})}}=p_0 + \alpha^{1/2} q_0\quad;\quad 2\,g(z_2)\sqrt{\f{\alpha}{f^{\prime\prime}(z_{2})}}=p_0 - \alpha^{1/2} q_0,
\ee
and these two equations can be solved with respect to $p_0$ and $q_0$ respectively to obtain
\be 
p_0 = \left(\f{\alpha}{d}\right)^{\f 14} \left(z_1^{3/2}+z_2^{3/2}\right) \quad;\quad q_0 =  \left(\f{1}{\alpha d}\right)^{\f 14} \left(z_1^{3/2}-z_2^{3/2}\right).                                                              
\label{eq:dp_44}
\ee 
Here, we have set $d=1-4t$.\\

\noindent Inserting \eqref{eq:dp_44} into \eqref{eq:H_series}, we arrive at an asymptotic expression for $H(t)$ in terms of the Airy function
\be 
\Ai(z) = \int_C \exp\left(\f{w^3}{3}-zw\right)dw.
\ee
\noindent This presents the main result of this section, 
\begin{lemma}{}

For complex $t$,
\begin{multline}
H(t) \sim \f{(q;q)_\infty}{2\pi i}\left(\f{1}{\alpha d}\right)^{\f{1}{4}}\exp\left(\disp\f{\beta}{\epsilon}\right) \left(\alpha^{1/2}\epsilon^{1/3}(z_1^{3/2}+z_2^{3/2})\Ai(\alpha\epsilon^{-2/3})+\right.\\
\left.+\epsilon^{2/3}(z_2^{3/2}-z_1^{3/2})\Ai^\prime(\alpha\epsilon^{-2/3})\right)\quad\m{as}~\epsilon \to 0^+.
\label{eq:dp_45}
\end{multline}

\label{Lemma:dp_05}
\end{lemma}

\subsection{Uniform asymptotic expansion of $H(qt)$}

The approach used in the last section can be applied in a completely analogous way to $H(qt)$. The function $f(z,t)$ remains the same, whereas now
\be 
	g(z) = \left(\f{1}{z(1-z)}\right)^{1/2}.
	\label{eq:dp_46}
\ee 
This changes the leading coefficients of the expansion towards
\be 
p_0 = \left(\f{\alpha}{d}\right)^{\f 14} \left(z_1^{1/2}+z_2^{1/2}\right) \quad;\quad q_0 =  \left(\f{1}{\alpha d}\right)^{\f 14} \left(z_1^{1/2}-z_2^{1/2}\right),
\label{eq:dp_47}
\ee 
and we obtain

\begin{lemma}{}

For complex $t$,
\begin{multline}
H(qt) \sim \f{(q;q)_\infty}{2\pi i}\left(\f{1}{\alpha d}\right)^{\f{1}{4}}\exp\left(\disp\f{\beta}{\epsilon}\right) \left(\alpha^{1/2}\epsilon^{1/3}(z_1^{1/2}+z_2^{1/2})\Ai(\alpha\epsilon^{-2/3})+\right.\\
\left.+\epsilon^{2/3}(z_2^{1/2}-z_1^{1/2})\Ai^\prime(\alpha\epsilon^{-2/3})\right)\quad\m{as}~\epsilon \to 0^+.
\label{eq:dp_48}
\end{multline}

\label{Lemma:dp_06}
\end{lemma}

\subsection{Uniform asymptotics of $G(t,q)$}

Combining Lemmas \ref{Lemma:dp_05} and \ref{Lemma:dp_06}, we arrive at
\begin{proposition}

For complex $t$ and $q\to 1^-$,
\be 
G(t,q) \sim \f{\alpha^{1/2}(z_1^{1/2}+z_2^{1/2})\Ai(\alpha\epsilon^{-2/3})+(z_2^{1/2}-z_1^{1/2})\epsilon^{1/3}\Ai^\prime(\alpha\epsilon^{-2/3})}{\alpha^{1/2}(z_1^{3/2}+z_2^{3/2})\Ai(\alpha\epsilon^{-2/3})+(z_2^{3/2}-z_1^{3/2})\epsilon^{1/3}\Ai^\prime(\alpha\epsilon^{-2/3})},
\label{eq:dp_49}
\ee
where $\epsilon = -\ln(q)$,
\be 
\nonumber
\left. \begin{array}{rcl} 
z_1 &=&\f{1}{2}(1+\sqrt{1-4t})\vspace{2mm}\\

z_2  &=&\f{1}{2}(1-\sqrt{1-4t})
\end{array} \right\}
\label{eq:dp_50}
\ee
and $\alpha(t)$ is given by Eq.\eqref{eq:dp_40}.
\label{prop:dp_01}

\end{proposition}

\begin{figure}[ht]
%\centering \epsfig{file=uniform_asymptotics_q_099.eps,width=0.7\textwidth}
\centering \includegraphics[width=0.7\textwidth]{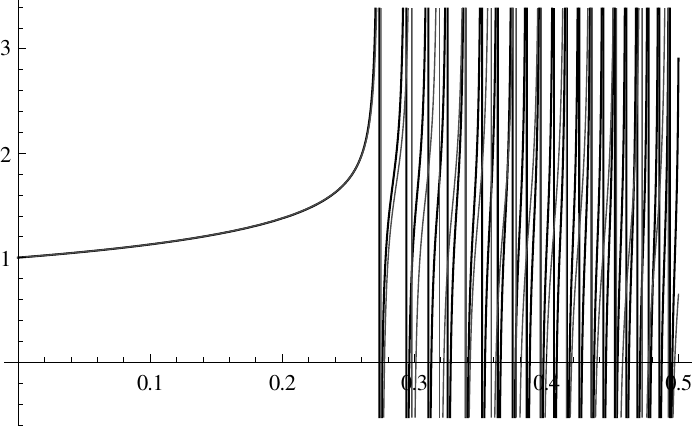}
\caption{Plot of $G(t,q)$ (black) against the uniform asymptotic expression \eqref{eq:dp_49} (grey) for $\epsilon=10^{-2}$ and $t$ ranging between 0 and 1/2 (horizontal axis).}
\label{fig:dp_04}
\end{figure}
One easily shows that for $t\leq t_c$ and $q\to 1^-$, $G(t,q)$ tends towards the generating function of the Catalan numbers,
\be 
C(t) = \f{1}{2t}\left(1-\sqrt{1-4t}\right).
\ee 
This is consistent with the well-known result for the generating function of unweighted Dyck paths. By applying Dini's theorem, one can further show that the convergence is \emph{uniform} for $t\in[\,0,t_c\,]$. However, (\ref{eq:dp_49}) is also valid for $t>t_c$, though not in a uniform sense due to the occurence of poles in the denominator. This fact is illustrated in \fig{fig:dp_04}, where we have plotted $G(t,q)$ against the uniform asymptotic expression \eqref{eq:dp_49} for $\epsilon=10^{-2}$ and $t \in [0,1/2]$. \\

\noindent Defining the tricritical scaling function
\be 
F(s):=\f{\Ai^{\prime}(s)}{\Ai(s)},
\ee
we can conclude
\begin{corollary} 
For fixed $s=(1-4t)\cd(1-q)^{-\phi}$ and $q\to 1^-$,
\be 
	G\left(t,q\right) \sim 2\left[1+(1-q)^{\disp{-\gamma_0}} F((1-4t)(1-q)^{-\disp\phi})\right],
\label{eq:dp_51}
\ee
where the critical exponents are $\phi=2/3$ and $\gamma_0=-1/3$.
\label{Corollary:dp_02}
\end{corollary}
\noindent In particular,
\be 
G(t_c,q) \sim 2\cd\left(1+A_0\cd(1-q)^{\disp-\gamma_0}\right)\qquad(q\to 1^-),
\label{eq:dp_52}
\ee 
where $A_0=\Ai^{\prime}(0)/\Ai(0)=-0.72\dots .$\\

Both \eqref{eq:dp_49} and \eqref{eq:dp_51} can be rearranged in order to obtain an asymptotic expression for $F(s)$. In \fig{fig:dp_05}, we have plotted $-2\,F(s)$ against the expressions obtained from \eqref{eq:dp_49} and \eqref{eq:dp_51} for $\epsilon=10^{-3},~10^{-4}$ and $10^{-5}$. For fixed value of $q$, \eqref{eq:dp_49} provides a more accurate approximation than \eqref{eq:dp_51}. 
\begin{figure}[htb]
\centering   \subfigure[]{
    \label{fig:dp_05a}
\includegraphics[width=0.7\textwidth]{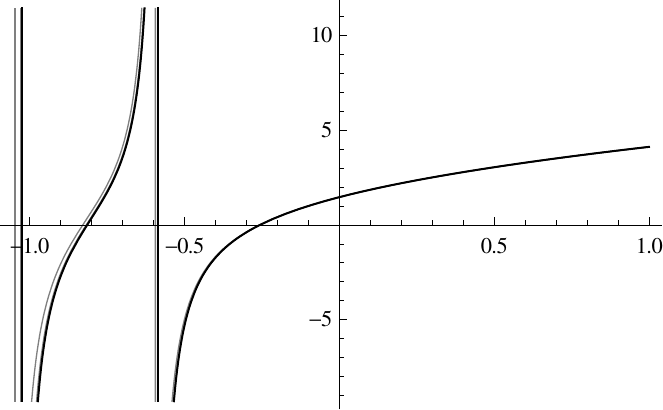}
  }\hspace{0.5cm}
\centering   \subfigure[]{
    \label{fig:dp_05b}
\centering \includegraphics[width=0.7\textwidth]{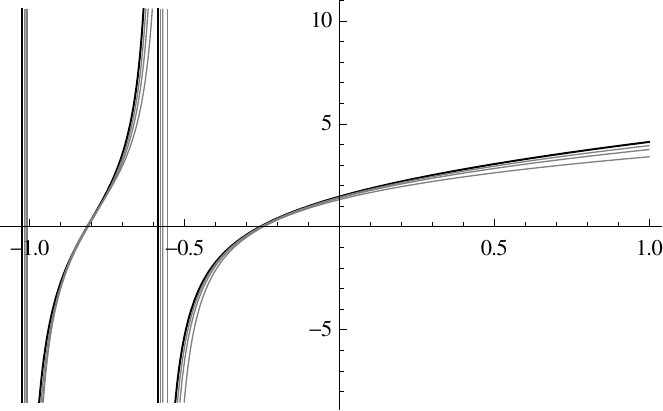}
}
  \caption{Plot of the scaling function $F(s)$ (black) against the asymptotic expressions \eqref{eq:dp_49} (a) and \eqref{eq:dp_51} (b) (grey) for $\epsilon=10^{-3},~10^{-4}$ and $10^{-5}$ (the smallest value of $\epsilon$ corresponds to the closest approximation).}
  \label{fig:dp_05}
\end{figure}

\subsection{Finite size scaling and specific heat}

Given the observed scaling behaviour \eqref{eq:dp_51} of the generating function, we now aim to calculate an expression for the finite size scaling function of the fixed-area partition function \eqref{eq:dp_015}. For this purpose, we first derive a further asymptotic expression for $G(t,q)$.\\

\noindent We define the \emph{singular part} of $G(t,q)$ as
\be 
G_{sing}(t,q) := G(t,q)-\f{1}{2t}
\ee
and prove
\begin{proposition}
Let $I=[t_0,t_c]$, where $0<t_0\leq t_c$. Then
\be 
G_{sing}(t,q) \sim \f{1}{2t}(1-q)^{-\disp\gamma_0}F((1-4t)(1-q)^{-\disp\phi})\quad(q\to 1^-),
\label{eq:dp_53}
\ee
uniformly for $t\in I$.
\label{prop:dp_02}
\end{proposition}
\begin{proof}
By using that $F(s) \sim -\sqrt{s}$ for $s\to\infty$, one easily sees that for $t\in I$, the RHS of \eqref{eq:dp_53} converges in a pointwise sense towards $C_{sing}:=-\sqrt{1-4t}/2t$, which is a continuous function on $I$. On also proves easily that for all $t\in I$, the RHS of \eqref{eq:dp_53} decreases monotonically with $q$. It therefore follows from Dini's theorem that the RHS converges uniformly to $C_{sing}(t)$. It is also clear that the same holds for the LHS, $G_{sing}(t,q)$. Hence, both sides of \eqref{eq:dp_53} have the same uniform asymptotic expression and are therefore uniformly asymptotic to each other.

\end{proof}

\noindent It is possible to use the Hadamard product expression \cite{Flajolet01}
\be 
\Ai(s)=\Ai(0)\exp(-A_0 s)\prod_{k=1}^\infty\left(1-\f{s}{s_k}\right),
\label{eq:dp_54}
\ee
\noindent where $s_k$ is the $k$th zero of the Airy function. Taking the derivative of the logarithm of this expression and exercising some careful analysis we get
\be 
F(s)=-\f{1}{s}\sum_{j=1}^\infty Z(j)s^j,
\label{eq:dp_55}
\ee 
where we have used the Airy Zeta function
\be 
Z(j) =\sum_{k=1}^\infty \left(\f{1}{s_k}\right)^j
\ee  
and inserted the conjectured value $Z(1)=-\Ai^\prime(0)/\Ai(0)$~\cite{Crandall96}.\\%=\disp{\f{1}{2\pi}}3^{5/6}\Gamma(\f{2}{3})^2$

\noindent Inserting \eqref{eq:dp_55} into \eqref{eq:dp_53}, we obtain that for $m\to\infty$,
\be
G_{sing}(t,q) &\sim& -\sum_{j=0}^\infty Z(j+1)(1-4t)^{j}(1-q)^{-2/3j+1/3}\nonumber\\
&=&-\sum_{m=0}^\infty\sum_{j=0}^\infty Z(j+1)(1-4t)^{j} \left({m+\f23 j-\f43}\atop{m}\right)q^m.
\label{eq:dp_56}
\ee
For $n\in \mathbb{N}$ and $\alpha\in\mathbb{C}\setminus\mathbb{Z}_{\leq 0}$, theorem VI.1 from \cite{Flajolet09} states that
\be
[z^n](1-z)^{-\alpha} = \left({n+\alpha-1}\atop{n}\right)\sim \f{n^{\alpha-1}}{\disp\Gamma(\alpha)}\left(1+\f{\alpha(\alpha-1)}{2n}\right).
\label{eq:dp_57}
\ee
\noindent Inserting the leading order of this expansion into \eqref{eq:dp_56} and extracting the $m$th coefficient, one formally arrives at
\be 
Q_m(t)\sim m^{-4/3}\sum_{j=0}^\infty \f{Z(j+1)}{\Gamma(\f23 j-\f13)}m^{2/3j}(1-4t)^j\quad(m\to\infty).
\label{eq:dp_58}
\ee
\noindent Defining the finite size scaling function
\BE
\phi(s):= \sum_{j=0}^\infty \f{Z(j)}{\Gamma(\f23 j-\f13)} s^j,
\label{eq:dp_59}
\EE
we can rewrite \eqref{eq:dp_58} as
\be 
Q_m(t) \sim m^{-4/3} \phi((1-4t)m^{2/3}).
\label{eq:dp_60}
\ee 
This expression is of the generic form expected for models which exhibit tricritical scaling \cite{Brak95,VanRensburg00}.\\

\section{Conclusion}

We have calculated an asymptotic expression for the generating function of Dyck paths, weighted with respect to both their perimeter and their area in the limit of the area generating variable tending towards 1. The result is valid uniformly for a range of values of the perimeter generating variable, including the tricritical point.

In the limit of both the perimeter and the area generating variable tending towards their critical values, we have shown the existence of a scaling function, expressible via Airy functions and their derivatives. The same type of scaling expression has been proven before to hold in the case of staircase polygons \cite{Prellberg95}.

Note in particular that the scaling function is obtained as a particular limit of the uniform asymptotic expansion. This is in contrast to the behaviour found in \cite{Prellberg95}, where the scaling function and the uniform asymptotic expansion are related by a local variable transformation. In \cite{Brak95}, uniform asymptotic expansions for tri-critical phase transitions were also constructed from scaling functions using such transformations.

From the scaling function of the generating function we derived an expression for the scaling function of the finite-area partition function.

\end{document}